\documentclass[11pt,letterpaper]{article}
\usepackage{times}
\usepackage{amstext,amsgen,latexsym,amsmath}
\usepackage{amssymb,amsfonts}
\usepackage[left=1in,right=1in,top=1in,bottom=1in]{geometry}

\newtheorem{theorem}{Theorem}
\newtheorem{definition}[theorem]{Definition}
\newtheorem{lemma}[theorem]{Lemma}

\newtheorem{claim}[theorem]{Claim}
\newcommand{\ket}[1]{\vert #1 \rangle}
\newcommand{\bra}[1]{\langle #1 \vert}
\newcommand{\qed}{$\square$}
\newenvironment{proof}{%
  \vspace{2ex}
  \noindent{\it Proof.\ }}{%
  \hspace*{\fill}\qed
  \vspace{2ex}}
\title{
Quantum Algorithms for Finding Constant-sized Sub-hypergraphs 
  }

\author{
{\large \hspace*{-1ex} $\text{Fran{\c c}ois Le Gall}^{\ast}$
 \hspace*{-1ex}}\\
\and
{\large \hspace*{-1ex} $\text{Harumichi Nishimura}^{\dagger}$
 \hspace*{-1ex}}\\
\and
{\large \hspace*{-1ex} $\text{Seiichiro Tani}^{\ddagger}$
 \hspace*{-1ex}}\\
}

\date{}

\begin{document}
\maketitle
\vspace*{-5mm}
\begin{center}
  ${}^{\ast}$Graduate School of Information Science and Technology, The University of Tokyo\\
{\tt legall@is.s.u-tokyo.ac.jp}\\
  ${}^{\dagger}$Graduate School of Informatics, Nagoya University\\
{\tt hnishimura@is.nagoya-u.ac.jp}\\
${}^{\ddagger}$NTT Communication Science Laboratories, NTT Corporation\\
{\tt tani.seiichiro@lab.ntt.co.jp}
\end{center}
\begin{abstract}
We develop a general framework to construct quantum algorithms 
that detect if a $3$-uniform hypergraph given as input contains a 
sub-hypergraph isomorphic to a prespecified constant-sized hypergraph. 
This framework is based on the concept of nested quantum walks 
recently proposed by Jeffery, Kothari and Magniez [SODA'13], and
extends the methodology designed by 
Lee, Magniez and Santha [SODA'13] for similar problems over graphs.
As applications, we obtain a quantum algorithm
for finding a $4$-clique in a $3$-uniform hypergraph on $n$ vertices
with query complexity $O(n^{1.883})$,
and a quantum algorithm for determining if a~{ternary operator} over a set of size $n$ is associative
with query complexity $O(n^{2.113})$. 
\end{abstract}

%\bigskip
%\centerline{{\bf Keywords}: Quantum algorithms, hypergraphs, quantum walk.}

\sloppy

%%%%%%%%%%%%%%%%%%%%%%%%%%%%%%%%%
\section{Introduction}
%%%%%%%%%%%%%%%%%%%%%%%%%%%%%%%%%

Quantum query complexity is a model of quantum computation, in which
the cost of computing a function is measured by the number of queries
that are made to the input given as a black-box.
In this model, it was exhibited in the early stage of quantum computing research
that there exist  quantum algorithms superior to the classical counterparts, 
such as Deutsch and Jozsa's algorithm, Simon and Shor's period finding algorithms, and Grover's search algorithm. 
Extensive studies following them have invented
a lot of powerful upper bound (i.e., algorithmic) techniques such as variations/generalizations of Grover's search algorithm or quantum walks. 
Although these techniques give tight bounds for many problems, 
there are still quite a few  cases for which no tight bounds are known.  Intensively studied problems among them are the $k$-distinctness problem~\cite{Amb07,BelovsFOCS12,Belovs+ICALP13} and the triangle finding problem~\cite{BelovsSTOC12,Buhrman+SIAM05,Jeffery+SODA13,Lee+SODA13,MSS07}.

A recent breakthrough is the concept of learning graph
introduced by Belovs~\cite{BelovsSTOC12}.
This concept enables one to easily find a special form of feasible solutions
to the minimization form (i.e., the dual form) of the general adversary bound~\cite{HLS07,ReichardtFOCS09},
and makes possible to detour the need of solving a semidefinite program of
exponential size to find a non-trivial upper bound. Indeed, Belovs~\cite{BelovsSTOC12} improved the long-standing $O(n^{13/10})$
upper bound~\cite{MSS07} of the triangle finding problem to $O(n^{35/27})$. 
His idea was generalized by Lee, Magniez and Santha~\cite{LMS12}
and Zhu~\cite{Zhu12} to obtain a quantum algorithm that finds
a constant-sized subgraph with complexity $o(n^{2-2/k})$, 
improving the previous best bound $O(n^{2-2/k})$~\cite{MSS07},
where $k$ is the size of the subgraph.
Subsequently, Lee, Magniez and Santha~\cite{Lee+SODA13} constructed 
a triangle finding algorithm with quantum query complexity
$O(n^{9/7})$. This bound was later shown by Belovs and Rosmanis~\cite{Belovs+CCC13} 
to be the best possible bound attained by the family of quantum algorithms 
whose complexities depend 
only on the index set of 1-certificates.
Ref.~\cite{Lee+SODA13} also gave a framework of 
quantum algorithms for finding a constant-sized subgraph,
based on which they showed that associativity testing (testing if a binary operator 
over a domain of size~$n$ is associative) 
has quantum query complexity~$O(n^{10/7})$.

Recently, Jeffery, Kothari and Magniez~\cite{Jeffery+SODA13} cast
the idea of the above triangle finding algorithms into the framework of quantum walks
(called nested quantum walks) by recursively performing the quantum walk algorithm given
by Magniez, Nayak, Roland and Santha~\cite{MNRS11} (which extended two seminal works 
for quantum walk algorithms, Szegedy's algorithm~\cite{SzegedyFOCS04} based on Markov chain and Ambainis' algorithm~\cite{Amb07} for $k$-element distinctness).
Indeed, they presented two quantum-walk-based triangle finding algorithms
of complexities $\tilde {O}(n^{35/27})$ and $\tilde{O}(n^{9/7})$, respectively.
The nested quantum walk framework was further employed in~\cite{Belovs+ICALP13} (but in a different way from~\cite{Jeffery+SODA13})
to obtain $\tilde{O}(n^{5/7})$ complexity for the $3$-distinctness problem.
This achieves the best known upper bound (up to poly-logarithmic factors), 
which was first obtained with the learning-graph-based approach~\cite{BelovsFOCS12}.

The triangle finding problem also plays a central role in several areas beside query complexity,
and it has been recently discovered that faster algorithms for (weighted versions of) 
triangle finding would lead to faster algorithms for
matrix multiplication~\cite{LeGallSODA12,Williams+FOCS10}, 
the 3SUM problem \cite{Williams+STOC09}, and
for Max-2SAT \cite{Williams05,WilliamsPhD07}. In particular,
Max-2SAT over $n$ variables is reducible
to finding a triangle with maximum weight over $O(2^{n/3})$ vertices;
in this context, although the final goal is a \emph{time-efficient} classical or quantum algorithm that finds a triangle with maximum weight, studying triangle finding in the query complexity model is a first step toward this goal.

{\bf Our results.} 
Along this line of research, this paper studies the problem of 
finding a 4-clique (i.e., the complete $3$-uniform hypergraph with $4$ vertices) in a $3$-uniform hypergraph,
a natural generalization of finding a triangle in an ordinary graph (i.e.,
a $2$-uniform hypergraph).
Our initial motivation comes from the complexity-theoretic importance of the problem.
Indeed, while it is now well-known that  
Max-3SAT over~$n$ variables is reducible to finding a $4$-clique 
with maximum weight in a 3-uniform hypergraph 
of $O(2^{n/4})$ vertices (the reduction is similar to the reduction from Max-2SAT
to triangle finding mentioned above; we refer to \cite{WilliamsPhD07} for details), 
no efficient classical algorithm for $4$-clique 
finding has been discovered so far. Constructing query-efficient algorithms
for this problem can be seen as a first step to investigate the possibility of 
faster (in the time complexity setting) classical or quantum algorithms for Max-3SAT.

Concretely, and more generally,
this paper gives a framework based on quantum walks
for finding any constant-sized sub-hypergraph in a $3$-uniform hypergraph 
({Theorem \ref{th:main}}). 
This is an extension of the learning-graph-based algorithm
in~\cite{Lee+SODA13} to the hypergraph case in terms of 
a nested quantum walk~\cite{Jeffery+SODA13}.
We illustrate this methodology by constructing a quantum algorithm that finds
a $4$-clique in a $3$-uniform hypergraph\footnote{We stress that, while this quantum algorithm can also be used to 
find with the same complexity a $4$-clique of maximal weight, this does not currently lead to a better algorithm for Max-3SAT since our algorithm is only query-efficient.}  
with query complexity $\tilde{O}(n^{241/128})=O(n^{1.883})$,
while na\"{i}ve Grover search over the $\binom{n}{4}$ combinations of vertices
only gives $O(n^{2})$.
As another application, we also construct a quantum algorithm 
that determines if a {ternary operator} is associative using 
$\tilde{O}(n^{169/80})=O(n^{2.113})$ queries,
while na\"{i}ve Grover search needs $O(n^{2.5})$ queries. 
  
In the course of designing the quantum walk framework,
we introduce several new technical ideas (outlined below) for analyzing nested quantum walks
to cope with difficulties that do not arise 
in the $2$-uniform case (i.e., ordinary graphs),
such as the fact that 
the size of the random subset taken in an inner walk
may vary depending on the random subsets taken in outer walks.
We believe that these ideas may be applicable to various problems beyond 
sub-hypergraph finding.

Our framework is another demonstration of the power of the concept of nested quantum walks, 
and of its wide applicability. In particular, we crucially rely on the high-level description and analysis made possible by the nested quantum walk formalism 
to overcome all the technical difficulties that arise when considering $3$-uniform hypergraphs.

{\bf Technical contribution.} 
Roughly speaking, the subgraph finding algorithm by  Lee, Magniez and Santha~\cite{Lee+SODA13} works as follows.
First, for each vertex $i$ in the subgraph $H$
that we want to find, 
a random subset $V_{i}$ of 
vertices of the input graph is taken.
This subset 
$V_i$ represents a set of candidates for the vertex $i$.
Next, for each edge $(i,j)$ in the subgraph $H$, 
a random subset of 
pairs in 
$V_{i}\times V_{j}$ is taken, representing a set of candidates for the edge $(i,j)$.
The most effective feature of their algorithm
is to introduce a parameter for each ordered pair $(V_{i},V_{j})$ 
that controls the average degree of a vertex in the bipartite graph between $V_i$ and $V_j$. 
To make the algorithm efficient, 
it is crucial to keep the degree of every vertex in $V_{i}$ almost equal to the 
value specified by the parameter. For this, they carefully devise
a procedure for taking pairs from $V_i\times V_j$. 

Our basic idea is similar in that we first, for each vertex $i$ in the sub-hypergraph $H$ that we want to find, 
take a random subset $V_{i}$ of vertices in 
the input $3$-uniform hypergraph as a set of candidates for the vertex $i$ and then,
for each hyperedge $\{i,j,k\}$ of $H$, take a random subset of triples in  $V_{i}\times V_{j} \times V_{k}$.
One may think that the remaining task is to fit the pair-taking procedure
into the hypergraph case. It, however, turns out to be technically very complicated to 
generalize the pair-taking procedure from \cite{Lee+SODA13} to an efficient triple-taking procedure.
Instead we cast the idea into the nested quantum walk of Jeffery, Kothari and Magniez~\cite{Jeffery+SODA13}
and employ probabilistic arguments. More concretely, 
we introduce  a parameter that specifies the number $e_{ijk}$ of triples
to be taken from $V_{i}\times V_{j} \times V_{k}$
for each hyperedge $\{i,j,k\}$ of $H$. 
We then argue that,
for randomly chosen $e_{ijk}$ triples,
the degree of each vertex sharply concentrates around its average, 
where the degree means the number of triples including the vertex 
(in this sense,  the parameters $e_{ijk}$ play essentially the same
role as those of ``average degrees'' used in  
~\cite{Jeffery+SODA13}, but introducing $e_{ijk}$ gives a neat formulation of the algorithm and this is effective particularly in handling such complicated cases as hypergraphs). 
This makes it substantially easier to analyze the complexity of all involved quantum walks,
and enables us to completely analyze the complexity of our approach.
Unfortunately, it turns out that this approach (taking the sets $V_i$ first, and then 
$e_{ijk}$ triples from each $V_{i}\times V_{j} \times V_{k}$) does not lead to any improvement 
over the na\"{i}ve $O(n^2)$-query quantum algorithm.

Our key idea is to introduce, for each unordered pair $\{i,j\}$ of vertices in $H$,  
a parameter $f_{ij}$, and modify the approach as follows.
After randomly choosing $V_{i},V_{j}, V_{k}$, we take three random subsets
$F_{ij}\subseteq V_{i}\times V_{j}$, $F_{jk}\subseteq V_{j}\times V_{k}$, and $F_{ik}\subseteq V_{i}\times V_{k}$
of 
size $f_{ij}$, $f_{jk}$ and $f_{ik}$, 
respectively. 
We then randomly choose $e_{ijk}$ triples from 
the set $\Gamma_{ijk}=\{(u,v,w)\:|\:(u,v)\in F_{ij},(u,w)\in F_{ik}\textrm{ and }(v,w)\in F_{jk}\}$.
The difficulty here is that the size of $\Gamma_{ijk}$ varies depending on the sets
$F_{ij}$, $F_{jk}$, $F_{ik}$.
Another problem is that, after taking many quantum-walks
(i.e., performing the update operation many times),
the distribution of the set of pairs can change.
To overcome these difficulties,
we carefully define the ``marked states'' 
(i.e., ``absorbing states'') of 
each level of the nested  quantum walk:
besides requiring, as usual, 
that the set (of the form $V_{i}$, $F_{ij}$ or $\Gamma_{ijk}$) associated to a marked state should contain 
a part (i.e., a vertex, a pair of vertices or a triple of vertices) of a copy of $H$,
we also require that this set should 
\emph{satisfy certain regularity conditions}.
We then show that the associated sets almost always satisfy the regularity conditions, by using  concentration theorems 
for hypergeometric distributions.
This regularity enables us to effectively bound the complexity of our 
new approach, giving in particular 
the claimed $\tilde{O}(n^{241/128})$-query upper bound when $H$ is a 4-clique.

%%%%%%%%%%%%%%%%%%%%%%%%%%%%%%%%%
\section{Preliminaries}\label{section_prelim}
%%%%%%%%%%%%%%%%%%%%%%%%%%%%%%%%%
For any $k\ge 2$, an undirected $k$-uniform hypergraph is a pair $(V,E)$,
where $V$ is a finite set (the set of vertices), and $E$ is a set of unordered $k$-tuples of 
elements in $V$ (the set of hyperedges). 
An undirected $2$-uniform hypergraph is simply an undirected graph.

In this paper, we use the standard quantum query complexity model formulated 
in Ref.~\cite{Beals+JACM01}. 
We deal with (undirected) 3-uniform hypergraphs $G=(V,E)$ 
as input, and the operation of the black-box is given as 
the unitary mapping $|\{u,v,w\},b\rangle\mapsto |\{u,v,w\},b\oplus \chi(\{u,v,w\})\rangle$ for $b\in \{ 0,1\}$,
where the triple $\{u,v,w\}$ is the query to the black-box 
and $\chi(\{u,v,w\})$ is the answer on whether the triple is a hyperedge of $G$, 
namely,  $\chi(\{u,v,w\})=1$ if $\{u,v,w\}\in E$ and $\chi(\{u,v,w\})=0$ otherwise. 

Our algorithmic framework is based on the concept of the nested quantum walk
introduced by Jeffery, Kothari and Magniez~\cite{Jeffery+SODA13}. 
In the nested quantum walk, for each positive integer $t$,
the walk at level $t$ checks whether
the current state is marked or not
by invoking 
the walk at level $t+1$, and this is iterated recursively until 
some fixed level~$m$.
The data structure of the walk at level $t$ 
is defined so that it includes the initial state of the walk at level $t+1$,
which means that
the setup cost of the walk at level~$t\ge 2$ is zero.
Jeffery, Kothari and Magniez have shown (in Section 4.1 of~\cite{Jeffery+SODA13}) that
the overall complexity of such a walk is
\[
\tilde O\left(
\mathsf{S}+\sum_{t=1}^m \left(\prod_{r=1}^t\frac{1}{\sqrt{\varepsilon_r}}\right)\frac{1}{\sqrt{\delta_t}}\mathsf{U}_t
\right)
\]
if the checking cost at level $m$ is zero, which will be our case.
Here $\mathsf{S}$ denotes the setup cost of the whole nested walk, $\mathsf{U_t}$ denotes the cost of updating the database 
of the walk at level $t$, $\delta_t$ denotes the spectral gap of the walk at level $t$,
and $\varepsilon_r$ denotes the fraction of marked states for the walk at level $r$.
As in most quantum walk papers, we only consider quantum walks on the Johnson graphs,
where the Johnson graph $J(N,K)$ is 
a graph such that each vertex is a subset 
with size $K$ of a set with size $N$ 
and two vertices corresponding to subsets $S$ and $S'$ 
are adjacent if and only if $|S\Delta S'|=2$ 
(we denote by $S\Delta S'$ the symmetric difference 
between $S$ and $S'$). 
If the walk at level $t$ is on $J(N,K)$,
then its spectral gap $\delta_{t}$ is known to be $\Omega(1/K)$.

Consider the update operation of the walk at any level.
The update cost may vary depending on the states of the walk we want to update. 
Assume without loss of generality
that the update operation is of the form
$U=\sum_{i}\ket{i}\bra{i}\otimes U_{i}$, where each $U_i$ can be implemented using 
$q_i$ queries,
and the quantum state to be updated is of the
form $\ket{s}=\sum_{i}\alpha_{i}\ket{i}\ket{s_{i}}$.
Then the following lemma, used in \cite{Jeffery+SODA13}, shows that 
if the magnitude of the states $\ket{i}\ket{s_{i}}$ 
that cost much to update (i.e., such that $q_i$ is large)
is small enough, we can approximate the update operator $U$
with good precision by replacing all $U_{i}$ acting on such costly states with
the identity operator.
\begin{lemma}[\cite{Jeffery+SODA13}]
\label{lm:ApproxUnitary}
Let $U=\sum_{i}\ket{i}\bra{i}\otimes U_{i}$ be a controlled unitary operator
and let $q_{i}$ be the query complexity of exactly implementing $U_{i}$.
For any fixed integer $T$, define $\tilde{U}$ 
as
$\sum_{i:q_{i}\le T}\ket{i}\bra{i}\otimes U_{i}+
\sum_{i:q_{i}>T}\ket{i}\bra{i}\otimes\Bbb{I}$,
where $\Bbb{I}$ is the identity operator on the space
on which $U_{i}$ acts.
Then,
 for any quantum state 
 $\ket{s}=\sum_{i}\alpha_{i}\ket{i}\ket{s_{i}}$,
the inequality
 $\left|\bra{s}\tilde{U}U\ket{s}\right|\ge 1-\epsilon_{T}$
holds
whenever  $ \epsilon_{T}\ge \sum_{i:q_{i}>T}\left|\alpha_{i}\right|^{2}$.
 \end{lemma}

In the analysis of this paper, 
hypergeometric distributions will appear many times.
Let $HG(n,m,r)$ denote the hypergeometric distribution whose random variable $X$ 
is defined by 
\[
\Pr[X=j]=\frac{\binom{m}{j}\binom{n-m}{r-j}}{\binom{n}{r}}.
\] 
We first state below several tail bounds of hypergeometric distributions 
(the proof can be easily obtained from Theorem 2.10 in \cite{JLR11}).

\begin{lemma}\label{lemma:HGtail}
When $X$ has a hypergeometric distribution with expectation $\mu$, 
the following hold (where $\exp(x)$ denotes $e^x$): 
\begin{itemize}
\item[(1)] 
For any $0<\delta\leq 1$, $\Pr[X\geq (1+\delta)\mu]\leq \exp(-\frac{\mu\delta^2}{3})$. %\\
\item[(2)]
For any $0<\delta<1$, $\Pr[X\leq (1-\delta)\mu]\leq \exp(-\frac{\mu\delta^2}{2})$. %\\
\item[(3)]
For any $\delta>2e-1$, 
$
\Pr[X>(1+\delta)\mu]<\left(\frac{1}{2}\right)^{(1+\delta)\mu}.
$
\end{itemize}
\end{lemma}

%%%%%%%%%%%%%%%%%%%%%%%%%%%%%%%%%
\section{Statement of our main result}
%%%%%%%%%%%%%%%%%%%%%%%%%%%%%%%%%

In this section, we state our main result (an upper bound on the query complexity of finding 
a constant-sized sub-hypergraph in a 3-uniform hypergraph) 
in terms of loading schedules, which generalizes the concept of 
loading schedules for graphs introduced, in the
learning graph framework, by Lee, Magniez and Santha~\cite{Lee+SODA13},
and used in the framework of nested quantum walks by Jeffery, Kothari and Magniez~\cite{Jeffery+SODA13}.

Let $H$ be a 3-uniform hypergraph with $\kappa$ vertices.
We identify the set of vertices of $H$ with the set $\Sigma_1=\{1,\ldots,\kappa\}$.
We identify the set of hyperedges 
of $H$ with the set $\Sigma_3\subseteq\{\{1,2,3\},\{1,2,4\},\ldots,\{\kappa-2,\kappa-1,\kappa\}\}$.
We identify the set of (unordered) pairs of vertices 
included in at least one hyperedge
of $H$ with the set $\Sigma_2=\{\{i,j\} \:|\: \{i,j,k\}\in \Sigma_3 \text{ for some $k$}\}$.
A loading schedule for $H$ is defined as follows.

\begin{definition}
A {\em loading schedule for $H$} of length $m$ is a list $S=(s_1,\ldots,s_m)$ of $m$ elements such that the following three properties hold for all $t\in\{1,\ldots,m\}$:
%\begin{itemize}
%\item
(i) $s_t\in \Sigma_1\cup \Sigma_2\cup \Sigma_3$;
%\item
(ii) if $s_t=\{i,j\}$, then there exist $t_1,t_2\in\{1,\ldots,t-1\}$ such that $s_{t_1}=i$ and $s_{t_2}=j$;
%\item
(iii) if $s_t=\{i,j,k\}$, then there exist $t_1,t_2,t_3\in\{1,\ldots,t-1\}$ such that $s_{t_1}=\{i,j\}$, $s_{t_2}=\{i,k\}$ and $s_{t_3}=\{j,k\}$. 
%;
%\end{itemize}
A loading schedule $S$ is {\em valid} if no element of $\Sigma_1\cup \Sigma_2\cup \Sigma_3$ appears more than once and, 
for any $\{i,j,k\}\in \Sigma_3$, there exists an index $t\in\{1,\ldots,m\}$ 
such that $s_t=\{i,j,k\}$.
\end{definition}

We now introduce the concept of parameters associated to a loading schedule.
Formally, these parameters are functions of the 
variable $n$ representing the number of vertices of the input 3-uniform hypergraphs $G=(V,E)$. 
We will nevertheless, in a slight abuse of notation, consider that $n$ is fixed, and define them as integers (implicitly depending on $n$). 
\begin{definition}
Let $S=(s_1,\ldots,s_m)$ be a loading schedule for $H$ of length $m$. 
A {\em set of parameters for $S$} is a set of $m$ integers defined as follows: 
for each $t\in\{1,\ldots,m\}$,
\begin{itemize}
\item
if $s_t=i$, then the associated parameter is denoted by $r_i$ and satisfies $r_i\in\{1,\ldots,n\}$;
\item
if $s_t=\{i,j\}$, then the associated parameter is denoted by $f_{ij}$ and satisfies $f_{ij}\in\{1,\ldots,r_ir_j\}$;
\item
if $s_t=\{i,j,k\}$, then the associated parameter is denoted by $e_{ijk}$ and satisfies $e_{ijk}\in\{1,\ldots,r_ir_jr_k\}$.
\end{itemize}
The set of parameters is {\em admissible} if $r_i\ge 1$, $e_{ijk}\ge 1$,
$\frac{r_ir_j}{f_{ij}}\geq 1$, $\frac{f_{ij}f_{ik}f_{jk}/(r_ir_jr_k)}{e_{ijk}}\geq 1$, and  
the terms $\frac{n}{r_i}$, 
$\frac{f_{ij}}{r_i}$, $\frac{f_{ij}}{r_j}$, $\frac{f_{ij}f_{ik}}{r_ir_jr_k}$ 
are larger than $n^{\gamma}$ for some constant $\gamma>0$.
\end{definition}

Now we state the main result in terms of loading schedules. 

\begin{theorem}\label{th:main}
Let $H$ be any constant-sized 3-uniform hypergraph.
Let $S=(s_1,\ldots,s_m)$ be a valid loading schedule for $H$ with an admissible set of parameters. 
There exists a quantum algorithm that, given as input a 3-uniform hypergraph $G$ with $n$ vertices, 
finds a sub-hypergraph of $G$ isomorphic to $H$ (and returns ``no'' if there are no such sub-hypergraphs) 
with probability at least some constant,
and has query complexity
\[
\tilde O\left(
\mathsf{S}+\sum_{t=1}^m \left(\prod_{r=1}^t\frac{1}{\sqrt{\varepsilon_r}}\right)\frac{1}{\sqrt{\delta_t}}\mathsf{U_t}
\right),
\]
where $\mathsf{S}$, $\mathsf{U}_t$, $\delta_t$ and $\varepsilon_r$ are evaluated as follows:
\begin{itemize}
\item $\mathsf{S}=\sum_{\{i,j,k\}\in \Sigma_3} e_{ijk}$;
\item for $t\in\{1,\ldots, m\}$,
%\begin{itemize}
%\item 
(i) if $s_t=\{i\}$, then $\delta_t=\Omega(\frac{1}{r_i})$, 
$\varepsilon_t=\Omega(\frac{r_i}{n})$ and 
$
\mathsf{U}_t=
\tilde O\left(1+\sum_{\{j,k\}:
\{i,j,k\}\in \Sigma_3}\frac{e_{ijk}}{r_i}\right)
$; 
%\item 
(ii) if  $s_t=\{i,j\}$, then $\delta_t=\Omega(\frac{1}{f_{ij}})$, $\varepsilon_t=\Omega(\frac{f_{ij}}{r_ir_j})$ and
$
\mathsf{U}_t=\tilde O\biggl( 1+\sum_{k:
\{i,j,k\}\in \Sigma_3}\frac{e_{ijk}}{f_{ij}}\biggr)
$; 
%\item 
(iii) if $s_t=\{i,j,k\}$, then $\delta_t=\Omega(\frac{1}{e_{ijk}})$, $\varepsilon_t=\Omega(\frac{e_{ijk}r_ir_jr_k}{f_{ij}f_{ik}f_{jk}})$ and 
$
\mathsf{U}_t=O(1).
$
%\end{itemize}
\end{itemize}
\end{theorem}

%%%%%%%%%%%%%%%%%%%%%%%%%%%%%%%%
\section{Proof of Theorem \ref{th:main}}
%%%%%%%%%%%%%%%%%%%%%%%%%%%%%%%%

In this section, we prove Theorem~\ref{th:main} 
by constructing an algorithm based on the concept of $m$-level nested quantum walks,
in which the  walk at level $t$ will correspond to the element $s_t$ of the loading schedule
for each 
$t\in\{1,\ldots,m\}$.
For convenience, we will write $M_{ijk}=11\frac{f_{ij}f_{ik}f_{jk}}{r_ir_jr_k}$ for each 
$\{i,j,k\}\in\Sigma_3$.

%%%%%
\subsection{Definition of the walks}\label{definition_of_the_walks}
%%%%%
At level $t\in\{1,\ldots,m\}$, the quantum walk will differ according to the nature of $s_t$, so there are three cases to consider.
\vspace{2mm}

\noindent{\bf Case 1 \boldmath{[$s_t=i$]}:}
The quantum walk will be over the Johnson graph $J(n,r_i)$.
The space of the quantum walk will then be 
$
\Omega_t=\left\{T\subseteq \{1,\ldots,n\} \mid |T|=r_i\right\}.
$
A state of this walk is an element $R_t\in \Omega_t$.
\vspace{2mm}

\noindent{\bf Case 2 \boldmath{[$s_t=\{i,j\}$]}:}
The quantum walk will be over $J(r_i r_j,f_{ij})$.
The space of the quantum walk will then be 
$
\Omega_t=\left\{T\subseteq \{1,\ldots,r_i r_j\} \mid |T|=f_{ij}\right\}.
$
A state of this walk is an element $R_t\in \Omega_t$.
\vspace{2mm}

\noindent{\bf Case 3 \boldmath{[$s_t=\{i,j,k\}$]}:}
The quantum walk will be over $J(M_{ijk},e_{ijk})$.
The space of the quantum walk will then be 
$
\Omega_t=\left\{T\subseteq \{1,\ldots,M_{ijk}\} \mid |T|=e_{ijk}\right\}.
$
A state of this walk is an element $R_t\in \Omega_t$.

%%%%%
\subsection{Definition of the data structures of the walks}\label{sub:datastructure}
%%%%%
Let us fix an arbitrary ordering on the set $V\times V\times V$
of triples of vertices.
For any set $\Gamma\subseteq V\times V\times V$
and any $R\subseteq\{1,\ldots,|V|^3\}$, 
define the set $\mathsf{Y}(R,\Gamma)$ 
consisting of at most $|R|$ triples of vertices which are taken from $\Gamma$ 
by the process below.
\begin{itemize}
\item
Construct a list $\Lambda$ of all the triples in $V\times V\times V$ as follows:
first, all the triples in $\Gamma$ are listed in increasing order and, then, all the triples in $(V\times V\times V)\backslash \Gamma$
are listed in  increasing order. 
\item
For any $a\in\{1,\ldots,|V|^3\}$, 
let $\Lambda[a]$ denote the $a$-th triple of the list.
\item
Define $\mathsf{Y}(R,\Gamma)=\{\Lambda[a]\:|\:a\in R\}\cap \Gamma$.
\end{itemize}
The following lemma 
will be useful later in this section.
\begin{lemma}\label{lemma:construction}
Let $\Gamma$ and $\Gamma'$ be two subsets of $V\times V\times V$. 
Let $p$ and $r$ be any parameters such that $1\le r\le p\le |V|^3$.
There exists a permutation $\pi$ of $\{1,\ldots,p\}$ such that, if $R$ is 
a subset of $\{1,\ldots,p\}$ of size $r$
taken uniformly at random, then
\[
\Pr_{R}\left[
|\mathsf{Y}
(R,\Gamma)\Delta\mathsf{Y}
(\pi(R),\Gamma')|\le\! 
\frac{22r|\Gamma\Delta \Gamma'|}{p}\!+\!100\log n
\right]\!\ge\! 
1-2\left(\frac{1}{2}\right)^{\!\!\!\frac{11r|\Gamma\Delta\Gamma'|}{p}+50\log n}\!.
\]
\end{lemma}
\begin{proof}
Let $\Lambda$ and $\Lambda'$ be the lists obtained when using the construction for $\Gamma$ and $\Gamma'$, respectively.
Let us write 
\begin{align*}
&\Lambda_1=\{{\Lambda}[a]\:|\:1\le a\le p\}\cap \Gamma,\\
&\Lambda'_1=\{{\Lambda'}[a]\:|\:1\le a\le p\}\cap\Gamma'.
\end{align*}
We can show the following inequality.
\begin{claim}\label{claim0}
$
\left|\Lambda_1\Delta \Lambda'_1\right|\le2|\Gamma\Delta \Gamma'|.
$
\end{claim}
\begin{proof}
$\Lambda_1$ contains precisely the $|\Lambda_1\cap(\Gamma\cap\Gamma')|$ 
smallest (with respect to the increasing order) elements of $\Gamma\cap\Gamma'$,
while the other
$\left|\Lambda_1\cap (\Gamma\backslash \Gamma')\right|$ elements of $\Lambda_1$
are in $\Gamma\backslash\Gamma'$.
Similarly, $\Lambda'_1$ contains precisely the $|\Lambda'_1\cap(\Gamma\cap\Gamma')|$ 
smallest elements of $\Gamma\cap\Gamma'$, while the other
$\left|\Lambda'_1\cap (\Gamma'\backslash \Gamma)\right|$ elements of $\Lambda'_1$ are in $\Gamma'\backslash\Gamma$.
We can write
\begin{align*}
\left|\Lambda_1\Delta \Lambda'_1\right|&=\Big||\Lambda'_1\cap(\Gamma\cap\Gamma')|-|\Lambda_1\cap(\Gamma\cap\Gamma')|\Big| +\left|\Lambda_1\cap (\Gamma\backslash \Gamma')\right|+\left|\Lambda'_1\cap (\Gamma'\backslash \Gamma)\right|\\
&\le\Big||\Lambda'_1\cap(\Gamma\cap\Gamma')|-|\Lambda_1\cap(\Gamma\cap\Gamma')|\Big| +\left|\Gamma\backslash \Gamma'\right|+\left|\Gamma'\backslash \Gamma\right|.
\end{align*}
We have to consider two cases.

\noindent
{\bf Case 1:} $|\Lambda_1|=|\Lambda'_1|=p$

\noindent
Assume, without loss of generality, that $|\Lambda_1\cap(\Gamma\cap\Gamma')|\le |\Lambda'_1\cap(\Gamma\cap\Gamma')|$.
We have 
\[
|\Lambda_1\cap(\Gamma\cap\Gamma')|= p-|\Lambda_1\cap(\Gamma\backslash\Gamma')|\ge p- |\Gamma\backslash \Gamma'|
\]
and $|\Lambda'_1\cap(\Gamma\cap\Gamma')|\le p$.
Thus
\[
\Big||\Lambda'_1\cap(\Gamma\cap\Gamma')|-|\Lambda_1\cap(\Gamma\cap\Gamma')|\Big|\le
p-(p- |\Gamma\backslash \Gamma'|)= |\Gamma\backslash \Gamma'|,
\]
which gives $\left|\Lambda_1\Delta \Lambda'_1\right|\le2|\Gamma\Delta \Gamma'|$, as claimed.

\noindent
{\bf Case 2:} $\min(|\Lambda_1|,|\Lambda'_1|)<p$

\noindent
By symmetry, it suffices to show only the case where $|\Lambda'_1|\leq|\Lambda_1|$.
Since $|\Lambda'_1|<p$, we have $\Lambda'_1=\Gamma'$.
This implies that $|\Lambda_1\setminus\Lambda'_1|=|\Lambda_1\setminus\Gamma'|
\leq |\Gamma\setminus\Gamma'|$. Since $|\Lambda'_1|\leq |\Lambda_1|$, 
we have $|\Lambda'_1\setminus\Lambda_1|\leq |\Lambda_1\setminus\Lambda'_1|
\leq |\Gamma\setminus \Gamma'|$. 
Hence, $|\Lambda_1\Delta\Lambda'_1|\leq 2|\Gamma\setminus\Gamma'|
\leq 2|\Gamma\Delta\Gamma'|$ also holds
in this case.
\end{proof}

For any $a\in\{1,\ldots,\min(p,|\Gamma|)\}$ such that $\Lambda[a]$ is in $\Lambda'_1$, we set $\pi(a)=a'$, 
where $a'$ is the (unique) index in $\{1,\ldots,\min(p,|\Gamma'|)\}$ such that 
$\Lambda[a]=\Lambda'[a']$.
For all other $a\in\{1,\ldots, p\}$, we set $\pi(a)$ arbitrarily such that $\pi$ becomes a permutation of $\{1,\ldots,p\}$.

Let $R$ be any subset of $\{1,\ldots,p\}$. Define the subsets $S_R,S'_R\subseteq R$ as follows:
\begin{align*}
&S_R=\left\{a\in R\:|\:\Lambda[a]\in\Lambda_1\backslash\Lambda'_1\right\},\\
&S'_R=\left\{b\in \pi(R)\:|\:\Lambda'[b]\in\Lambda'_1\backslash\Lambda_1\right\}.
\end{align*}
From the definition of $\pi$, we know that for
any element $a\in R\cap\{1,\ldots,|\Gamma|\}$ such that $a\notin S_R$
we have $\Lambda[a]=\Lambda'[\pi(a)]$, 
which means that this element is not in 
$\mathsf{Y}(R,\Gamma)\backslash\mathsf{Y}(\pi(R),\Gamma')$.
This implies that
$
\mathsf{Y}(R,\Gamma)\backslash\mathsf{Y}(\pi(R),\Gamma')\subseteq\{\Lambda[a]\:|\:a\in S_R\}.
$
Similarly, we have \sloppy
$
\mathsf{Y}(\pi(R),\Gamma')\backslash\mathsf{Y}(R,\Gamma)\subseteq\{\Lambda'[a]\:|\:a\in S_R'\},
$ 
which gives the inequality
\begin{align*}
|\mathsf{Y}(R,\Gamma)\Delta\mathsf{Y}(\pi(R),\Gamma')|
&\le |S_R|+|S'_R|.
\end{align*}
Recall that $R$ is taken uniformly at random from $\{1,\ldots,p\}$ 
so that $|R|=r$.
Thus, $|S_R|$ has hypergeometric distribution $HG(p,|\Lambda_1\backslash\Lambda'_1|,r)$
and its expectation is  $\mu=\frac{r|\Lambda_1\backslash \Lambda'_1|}{p}$.
Taking 
$\delta=\frac{1}{\mu}(11\frac{r|\Gamma\Delta\Gamma'|}{p}+50\log n)-1$,
we have 
\[
\Pr\left[|S_R|\geq 11\frac{r|\Gamma\Delta\Gamma'|}{p}+50\log n\right]
=
\Pr\left[|S_R|\geq (1+\delta)\mu\right].
\]
Note that by Claim~\ref{claim0}, 
$
1+\delta=
\frac{1}{\mu}(11\frac{r|\Gamma\Delta\Gamma'|}{p}+50\log n)
>11\cdot \frac{|\Gamma\Delta\Gamma'|}{|\Lambda_1\Delta\Lambda'_1|}
\geq 11/2
$
and hence $\delta\geq 9/2>2e-1$. By Lemma~\ref{lemma:HGtail}(3), 
we have
\[
\Pr\left[|S_R|\geq 11\frac{r|\Gamma\Delta\Gamma'|}{p}+50\log n\right]
<
\bigl(1/2\bigr)^{(1+\delta)\mu}
=
\bigl(1/2\bigr)^{11\frac{r|\Gamma\Delta\Gamma'|}{p}+50\log n}.
\]
Similarly,
$
\Pr\left[|S'_R|\geq 11\frac{r|\Gamma\Delta\Gamma'|}{p}+50\log n\right]
<
\left(1/2\right)^{11\frac{r|\Gamma\Delta\Gamma'|}{p}+50\log n}.
$ 
Therefore,
\begin{align*}
&\!\!\!\!\!\!\Pr\left[|\mathsf{Y}(R,\Gamma)\Delta\mathsf{Y}(\pi(R),\Gamma')|
\geq 22\frac{r|\Gamma\Delta\Gamma'|}{p}+100\log n\right]\\
&\leq
\Pr\left[|S_R|+|S'_R|\geq 22\frac{r|\Gamma\Delta\Gamma'|}{p}+100\log n\right]\\
&<
2\bigl(1/2\bigr)^{11\frac{r|\Gamma\Delta\Gamma'|}{p}+50\log n}.
\end{align*}
This completes the proof of Lemma \ref{lemma:construction}.
\end{proof}

Suppose that the states of the walks at levels $1,\ldots,m$
are $R_{1},\ldots,R_{m}$, respectively. 
Assume that the set of vertices of $G$ is
$V=\{v_{1},\ldots,v_{n}\}$.
We first interpret the states  $R_{1},\ldots,R_{m}$ 
as sets of vertices, sets of pairs of vertices or sets of triples of vertices in~$V$,
as follows.
For each $t\in\{1,\ldots,m\}$, there are three cases to consider.
\vspace{2mm}

\noindent{\bf Case 1 \boldmath{[$s_t=i$]}:}
In this case, $R_t=\{a_1,\ldots,a_{r_i}\}\subseteq \{1,\ldots,n\}$.
We associate to $R_t$ the set $V_i=\{v_{a_1},\ldots,v_{a_{r_i}}\}$.
For further reference, we will rename the vertices in this set as 
$V_i=\{v^{i}_{1},\ldots,v^i_{r_i}\}$.
\vspace{2mm}

\noindent{\bf Case 2 \boldmath{[$s_t=\{i,j\}$ with $i<j$]}:}
We know that, in this case, there exist $t_1,t_2\in\{1,\ldots,t-1\}$ such that $s_{t_1}=i$ and $s_{t_2}=j$.
The state $R_t$ represents a set $\{(a_1,b_1),\ldots,(a_{f_{ij}},b_{f_{ij}})\}$ of $f_{ij}$ pairs in $R_{t_1}\times R_{t_2}$.
We associate to it the set $F_{ij}=\{(v^i_{a_1},v^j_{b_1}),\ldots,(v^i_{a_{f_{ij}}},v^j_{b_{f_{ij}}})\}$ of pairs
of vertices.
\vspace{2mm}

\noindent{\bf Case 3 \boldmath{[$s_t=\{i,j,k\}$ with $i<j<k$]:}}
We know that there exist $t_1,t_2,t_3\in\{1,\ldots,t-1\}$ such that $s_{t_1}=\{i,j\}$, $s_{t_2}=\{i,k\}$
and $s_{t_3}=\{j,k\}$, and $R_t$ is a subset of $\{1,\ldots,M_{ijk}\}$
with $|R_t|=e_{ijk}$. Let us define the set
$$
\Gamma_{ijk}=\left\{(u,v,w)\in V_i\times V_j\times V_k \:|\: (u,v)\in F_{ij}\textrm{, }(u,w)\in F_{ik}\textrm{ and }(v,w)\in F_{jk}\right\}.
$$
We associate to $R_t$ the set $E_{ijk}=\mathsf{Y}(R_t,\Gamma_{ijk})$.
\vspace{2mm}

We are now ready to define the data structures involved in the walks.
When the states of the walks at levels $1,\ldots,(m-1)$ are $R_{1},\ldots,R_{m-1}$, respectively, and the state of the most inner walk is $R_m$, the data structure associated with the most inner walk 
is denoted by $D(R_{1},\ldots,R_{m})$ and defined as:
\[
D(R_{1},\ldots,R_{m})=\biggl\{(\{u,v,w\},\chi(\{u,v,w\}))\:|\:
(u,v,w)
\in \bigcup_{\{i,j,k\}\in\Sigma_3\colon i<j<k}E_{ijk}\biggr\}.
\]
The data structure associated with the walk at level $t$, for each $t\in\{1,\ldots,m-1\}$, is defined as:
\[
\sum_{R_{t+1}\in \Omega_{t+1}}\cdots\sum_{R_m\in \Omega_{m}}
\ket{R_{t+1}}\cdots\ket{R_{m}}\ket{D(R_{1},\ldots,R_{m})}.
\]
Here and hereafter we omit normalization factors.
%%%%%%%
\subsection{Marked states of the walks}\label{subsection:marked}
%%%%%%%
For any $t\in\{1,\ldots,m-1\}$, the purpose of the walk at level $t+1$ is to check if the state of the walk $t$ is marked
(for the most inner walk, the state can be checked without running another walk, since all the information necessary is already in the database).
In this subsection we define the set of marked states for each walk.

Assume that the hypergraph $G$ contains a (without loss of generality, unique) 
sub-hypergraph isomorphic to $H$. 
Let $\{u_1,\ldots,u_{\kappa}\}$ denote the vertex set of this sub-hypergraph.  
For the most outer walk, $s_1=j$ for some $j\in\{1,\ldots,\kappa\}$ and we say that $R_1$ is marked if and only if $u_j\in V_j$.
Consider a state $R_t$ of the walk at level $t>1$, and suppose that the states $R_1,\ldots,R_{t-1}$ are all marked. 
We have again three cases to consider. 
\vspace{2mm}

\noindent{\bf Case 1 \boldmath{[$s_t=i$]}:}
$R_t$ corresponds to $V_i$. We say that $R_t$ is marked if and only if $u_i\in V_i$.
\vspace{2mm}

\noindent{\bf Case 2 \boldmath{[$s_t=\{i,j\}$ with $i<j$]}:}
$R_t$ corresponds to $F_{ij}$, and we say that $R_t$ is marked if and only if
the following four conditions hold:
\begin{itemize}
\item[(a)]
%(a) 
$(u_i,u_j)\in F_{ij}$;
\item[(b)]
%(b) 
for all $u\in V_i$,  $\frac{f_{ij}}{2r_i}\le|\{v\in V_j\:|\: (u,v)\in F_{ij}\}|\le 2\frac{f_{ij}}{r_i}$;
\item[(c)]
%(c) 
for all $v\in V_j$,  $\frac{f_{ij}}{2r_j}\le|\{u\in V_i\:|\: (u,v)\in F_{ij}\}|\le 2\frac{f_{ij}}{r_j}$;
\item[(d)]
%(d) 
for any $k$ such that there exists $t_1\in\{1,\ldots,t-1\}$ for which $s_{t_1}=\{i,k\}$,
and any $(v,w)\in V_j\times V_k$,
$|\{u\in V_i\:|\:(u,v)\in F_{ij} \textrm{ and }(u,w)\in F_{ik}\}|\le 11\frac{f_{ij}f_{ik}}{r_ir_jr_k}$.
\end{itemize}
\vspace{2mm}

\noindent{\bf Case 3 \boldmath{[$s_t=\{i,j,k\}$ with $i<j<k$]}:}
$R_t$ corresponds to $E_{ijk}$, and we say that $R_t$ is marked if and only if
$(u_i,u_j,u_k)\in E_{ijk}$.
\vspace{2mm}
%\vspace{0.2mm}

The next subsection will use the following lemma. 
\begin{lemma}\label{lemma:balanced}
Assume that, for Case 2, $R_t$ is taken uniformly at random from $\Omega_t$ 
(i.e., $F_{ij}$ corresponds to a set of $f_{ij}$ pairs taken 
uniformly at random from $V_i\times V_j$). Then,
\[
\Pr[\textrm{Conditions (b),(c),(d) hold for $F_{ij}$} ]\ge
1-2r_i e^{- \frac{f_{ij}}{8r_i}} -2r_j e^{-\frac{f_{ij}}{8r_j}}
-r_{j}r_k\kappa 2^{-11\frac{f_{ij}f_{ik}}{r_ir_jr_k}}.
\]
%\begin{align*}
%\Pr[\textrm{Conditions (b),(c),(d) hold for $F_{ij}$} ]\ge
%1&-2r_i\exp(- \frac{f_{ij}}{8r_i})-2r_j\exp(-\frac{f_{ij}}{8r_j})\\
%&-r_{j}r_k\kappa 2^{-11\frac{f_{ij}f_{ik}}{r_ir_jr_k}}.
%\end{align*}
\end{lemma}
\begin{proof}
For each $u\in V_i$, the quantity $|\{v\in V_j\:|\: (u,v)\in F_{ij}\}|$ is a random variable with hypergeometric distribution $HG(r_ir_j,r_j,f_{ij})$. 
Its expectation is $f_{ij}/r_i$, and by Lemma \ref{lemma:HGtail}(1-2) it holds that  
\begin{align*}
\Pr\left[|\{v\in V_j\:|\: (u,v)\in F_{ij}\}|> 2\frac{f_{ij}}{r_i}\right]&\le 
\exp(-\frac{1}{3}\times \frac{f_{ij}}{r_i})
\leq
\exp(-\frac{1}{8}\times \frac{f_{ij}}{r_i}),
\\
\Pr\left[|\{v\in V_j\:|\: (u,v)\in F_{ij}\}|< \frac{f_{ij}}{2r_i}\right]&\le \exp(-\frac{1}{8}\times \frac{f_{ij}}{r_i}).
\end{align*}
A similar statement holds for the degree of each $v\in V_j$,
and thus, from the union bound, we obtain
\[
\Pr[\textrm{Condition (b) or (c) does not hold for $F_{ij}$} ]\le 2r_i\exp(-\frac{f_{ij}}{8r_i})+2r_j\exp(-\frac{f_{ij}}{8r_j}).
\]  

For any $k$ such that there exists $t_1\in\{1,\ldots,t-1\}$ for which $s_{t_1}=\{i,k\}$,
consider any $(v,w)\in V_j\times V_k$. Let us write $S=\{u\in V_i\:|\:(u,w)\in F_{ik}\}$. 
Since $R_{t_1}$ is marked, we know that $|S|\le 2f_{ik}/r_k$.
The quantity $|\{u\in V_i\:|\:(u,v)\in F_{ij} \textrm{ and }(u,w)\in F_{ik}\}|$ has hypergeometric
distribution $HG(r_ir_j,|S|,f_{ij})$.  
Applying Lemma \ref{lemma:HGtail}(3)
with $\delta=\frac{11f_{ik}/r_k}{|S|}-1>2e-1$,
we obtain
\[
\Pr\left[
|\{u\in V_i\:|\:(u,v)\in F_{ij} \textrm{ and }(u,w)\in F_{ik}\}|\ge 11\frac{f_{ij}f_{ik}}{r_ir_jr_k}
\right]\le 2^{-11\frac{f_{ij}f_{ik}}{r_ir_jr_k}}.
\]
Using the union bound (note in particular that there are at most $\kappa$ possibilities for $k$), we conclude that
\[
\Pr\left[
\textrm{Condition (d) does not hold for $F_{ij}$}
\right]\le r_{j}r_k\kappa2^{-11\frac{f_{ij}f_{ik}}{r_ir_jr_k}}.
\]
The statement of the lemma then follows from the union bound.
\end{proof}

%%%
\subsection{Analysis of the algorithm}\label{sub:analysis}
%%%
Our nested quantum walk algorithm finds a marked state in the most inner walk and thus a sub-hypergraph isomorphic to $H$, with high probability, 
since, as will be shown below, the ideal nested quantum walks can 
be approximated with high accuracy. 
As explained in Section \ref{section_prelim},
the overall query complexity of the walk is 
\[
\tilde O\left(
\mathsf{S}+\sum_{t=1}^m \left(\prod_{r=1}^t\frac{1}{\sqrt{\varepsilon_r}}\right)\frac{1}{\sqrt{\delta_t}}\mathsf{U}_t
\right).
\]
We will show below that the values of the terms $\mathsf{S}$, $\mathsf{U_t}$, $\delta_t$ and $\varepsilon_t$ are 
as claimed in the statement of Theorem~\ref{th:main}.

We first make the following simple observation: when computing $\mathsf{U}_t$ and $\varepsilon_t$, we can assume that the state $R_{t-1}$ of the immediately outer walk is marked (and thus, by applying 
this argument recursively, that the states $R_1$,\ldots, $R_{t-1}$ of all the outer walks are marked). 
Indeed, remember that the purpose of the walk at level $t$ is to check if the state $R_{t-1}$ is marked.
We first evaluate its complexity under
the assumption that $R_{t-1}$ is marked, giving some upper bound $T$ on the complexity. 
Now, since the checking procedure in our framework has one-sided error, 
in the case where $R_{t-1}$ is not marked
the checking procedure may not terminate after $T$ queries, but we can stop it  after $T$ queries anyway and simply output that $R_{t-1}$ is not marked.

The setup cost $\mathsf{S}$ for the algorithm is the number of queries needed 
to construct the superposition
\[
\sum_{R_{1}\in \Omega_{1}}\cdots\sum_{R_m\in \Omega_{m}}
\ket{R_{1}}\cdots\ket{R_{m}}\ket{D(R_{1},\ldots,R_{m})}.
\]
This value is at most $\sum_{\{i,j,k\}\in \Sigma_3} e_{ijk}$.

We next evaluate 
$\delta_t$ and $\varepsilon_t$. The analysis is again divided into three cases.
\vspace{2mm}

\noindent{\bf Case 1 \boldmath{[$s_t=i$]}:}
Since the quantum walk is over $J(n,r_i)$ by the definition 
in Section~\ref{definition_of_the_walks}, we have $\delta_t=\Omega(\frac{1}{r_i})$ 
and $\varepsilon_t=\Omega(\frac{r_i}{n})$.
\vspace{2mm}

\noindent{\bf Case 2 \boldmath{[$s_t=\{i,j\}$ with $i<j$]}:}
Since the quantum walk is over $J(r_ir_j,f_{ij})$,
we have $\delta_t=\Omega(\frac{1}{f_{ij}})$. 
The fraction of states $F_{ij}$ such that $(u_i,u_j)\in F_{ij}$ is $\Omega(\frac{f_{ij}}{r_ir_j})$.
While all those states may not be marked, Lemma \ref{lemma:balanced} 
implies that the fraction of those states that are not marked is exponentially small when 
the set of parameters is admissible. 
Thus
$\varepsilon_t=\Omega(\frac{f_{ij}}{r_ir_j})$.
\vspace{2mm}

\noindent{\bf Case 3 [\boldmath{$s_t=\{i,j,k\}$ with $i<j<k$}]:}
In this case $\delta_t=\Omega(\frac{1}{e_{ijk}})$. 
\sloppy Since all the states $R_{1},\ldots,R_{t-1}$ of the outer walks are assumed 
to be marked, by item (d) of the definition of the marked states in Section \ref{subsection:marked}, we can upper-bound
$|\Gamma_{ijk}|
=
\sum_{(v,w)\in F_{jk}}|\{u\in V_i\mid (u,v)\in F_{ij}\mbox{ and }(u,w)\in F_{ik}\}|$
by $ |F_{jk}|%\times
\frac{11f_{ij}f_{ik}}{r_ir_jr_k}=M_{ijk}.
$
Thus, we have $\varepsilon_t=\Omega(\frac{e_{ijk}}{M_{ijk}})$.
\vspace{2mm}

Finally, we evaluate the cost $\mathsf{U}_t$,
which is the cost of transforming the quantum state
\[\sum_{R_{t+1}\in \Omega_{t+1}}\cdots\sum_{R_m\in \Omega_{m}}
\ket{R_{t+1}}\cdots\ket{R_{m}}\ket{D(R_{1},\ldots,R_{t-1},R_{t},R_{t+1},\ldots,R_{m})},
\]
to the quantum state
\[\sum_{R_{t+1}\in \Omega_{t+1}}\cdots\sum_{R_m\in \Omega_{m}}
\ket{R_{t+1}}\cdots\ket{R_{m}}\ket{D(R_{1},\ldots,R_{t-1},R'_{t},R_{t+1},\ldots,R_{m})},\]
for any two states $R_t$ and $R'_t$ adjacent in the corresponding Johnson graph.
We again divide the analysis into three cases.
\vspace{2mm}

\noindent{\bf Case 1 \boldmath{[$s_t=i$]}:} 
%\noindent
In this case $R_t$ and $R'_t$ are two subsets of $\{1,\ldots,n\}$, both of size $r_i$,
differing by exactly one element. The corresponding subsets $V_{i}$ and $V'_{i}$
also differ by exactly one element: let us write $V'_{i}=(V_i\backslash \{u\})\cup\{u'\}$. 
For any $\{i,j,k\}\in\Sigma_3$, there exist some $t_1,t_2,t_3\in\{t+1,\ldots,m\}$
such that $s_{t_1}=\{i,j\}$, $s_{t_2}=\{i,k\}$ and $s_{t_3}=\{i,j,k\}$. There also exist some
$t_4,t_5,t_6\in\{1,\ldots,m\}$ such that $s_{t_4}=j$, $s_{t_5}=k$ and $s_{t_6}=\{j,k\}$.
Note that $t_4,t_5,t_6$ can be smaller than $t$, but we will first assume 
that they are all larger than $t$ (the other cases, which are actually easier to analyze, 
are discussed at the end of the analysis). 
A state $R_{t_4}$ defines a set $V_j$ of $r_j$ vertices and, for any $R_{t_1}\in \Omega_{t_1}$,
the state $(R_t,R_{t_1},R_{t_4})$
defines a set of $f_{ij}$ pairs in $V_i\times V_j$, as described in Section~\ref{subsection:marked}.
In the same way, for any $R'_{t_1}\in \Omega_{t_1}$, 
the state $(R'_t,R'_{t_1},R_{t_4})$
defines a set of $f_{ij}$ pairs in $V'_i\times V_j$.
There exists a permutation $\pi_1$ of the elements of $\Omega_{t_1}$ such that,
for any $R_{t_1}\in \Omega_{t_1}$,
the set $F_{ij}$ defined by $(R_t,R_{t_1},R_{t_4})$ 
and the set $F'_{ij}$ defined by $(R'_t,\pi_1(R_{t_1}),R_{t_4})$ 
are related in the following way:
\[
F'_{ij}=(F_{ij}\backslash\{(u,v)\in\{u\}\times V_j\:|\:(u,v)\in F_{ij}\})\cup\{(u',v)\in\{u'\}\times V_j\:|\:(u,v)\in F_{ij}\},
\]
which means that each pair of the form $(u,v)$ in $F_{ij}$ is replaced by the pair $(u',v)$ in $F'_{ij}$, while the other 
pairs are the same in $F_{ij}$ and in $F'_{ij}$. 

\sloppy Similarly, there exists a permutation $\pi_2$ of the elements of $\Omega_{t_2}$ such that,
for any $R_{t_2}\in \Omega_{t_2}$,
the set $F_{ik}$ defined by $(R_t,R_{t_2},R_{t_5})$ 
and the set $F'_{ik}$ defined by $(R'_t,\pi_2(R_{t_2}),R_{t_5})$ 
are related in the following way:
\[
F'_{ik}=(F_{ik}\backslash\{(u,w)\in\{u\}\times V_k\:|\:(u,w)\in F_{ik}\})\cup\{(u',w)\in\{u'\}\times V_k\:|\:(u,w)\in F_{ik}\}.
\]

The states $(R_t,R_{t_1},R_{t_2},R_{t_4},R_{t_5},R_{t_6})$ define sets 
$V_{i},F_{ij},F_{ik},V_j,V_k,F_{jk},\Gamma_{ijk}$, while the
states $(R'_t,\pi_1(R_{t_1}),\pi_2(R_{t_2}),R_{t_4},R_{t_5},R_{t_6})$ define 
sets $V'_{i},F'_{ij},F'_{ik},V_j,V_k$, $F_{jk},\Gamma'_{ijk}$.
Given any state $R_{t_3}$,
let $E_{ijk}(R_{t},R_{t_1},R_{t_2},R_{t_3},R_{t_4},R_{t_5},R_{t_6})$ denote the set of hyperedges to be queried associated with $\Gamma_{ijk}$ and $R_{t_3}$,
and $E_{ijk}(R'_{t},\pi_1(R_{t_1}),\pi_2(R_{t_2}),R_{t_3},R_{t_4},R_{t_5},R_{t_6})$ denote the set of hyperedges to be queried associated with $\Gamma'_{ijk}$ and $R_{t_3}$.
By Lemmas~\ref{lm:ApproxUnitary} and \ref{lemma:construction}, 
the mapping
\begin{align*}\hspace{-5mm}
\ket{R_{t_1}}&\ket{R_{t_2}}\ket{R_{t_4}}\ket{R_{t_5}}\ket{R_{t_6}}\sum_{R_{t_3}\in\Omega_{t_3}}\ket{R_{t_3}}\ket{E_{ijk}(R_{t},R_{t_1},R_{t_2},R_{t_3},R_{t_4},R_{t_5},R_{t_6})}
\mapsto\\
&\hspace{-6mm}\ket{\pi_1(R_{t_1})}\ket{\pi_2(R_{t_2})}\ket{R_{t_4}}\ket{R_{t_5}}\ket{R_{t_6}}\!\!\sum_{R_{t_3}\in\Omega_{t_3}}\ket{R_{t_3}}\ket{E_{ijk}(R'_{t},\pi_1(R_{t_1}),\pi_2(R_{t_2}),R_{t_3},R_{t_4},R_{t_5},R_{t_6})}
\end{align*}
can be approximated within inverse polynomial 
precision\footnote{Note that a better estimation of the accuracy of the approximation can be obtained, but in this proof  approximation within inverse polynomial will be enough for our purpose. In consequence, while stronger tail bounds can be proved, the statements of Lemmas 
\ref{claim1} and \ref{claim2} will be enough for our purpose.} 
using 
$
\tilde O\left(\frac{e_{ijk}|\Gamma_{ijk}\Delta\Gamma'_{ijk}|}{M_{ijk}}+\log n\right) = 
\tilde O\left(\frac{e_{ijk}|\Gamma_{ijk}\Delta\Gamma'_{ijk}|}{M_{ijk}}+1\right)
$
queries (here Lemma~\ref{lm:ApproxUnitary} is used  
with 
 $T=22\frac{e_{ijk}|\Gamma_{ijk}\Delta\Gamma'_{ijk}|}{M_{ijk}}+100\log n$,
and then $\epsilon_{T}$ can be set to
$2\left(\frac{1}{2}\right)^{11\frac{e_{ijk}|\Gamma_{ijk}\Delta\Gamma'_{ijk}|}{M_{ijk}}+50\log n}$
by Lemma~\ref{lemma:construction}
with $p=M_{ijk}$ and $r=e_{ijk}$). 

We will use the following lemma. 
\begin{lemma}\label{claim1}
When $R_{t_1}, R_{t_2}$ and $R_{t_6}$ are taken uniformly at random,
\[
\Pr\left[
|\Gamma_{ijk}\Delta\Gamma'_{ijk}|\ge 44\times\frac{f_{ij}f_{ik}f_{jk}}{r_i^2r_jr_k}\right]=O\left(\frac{1}{n^{100}}\right).
\]
\end{lemma}
\begin{proof}
Let us write
\begin{align*}
A&=\{(u,v,w)\in\{u\}\times V_j\times V_k\:|\: (u,v)\in F_{ij} \textrm{, }(u,w)\in F_{ik}\textrm{ and }(v,w)\in F_{jk}\},\\
B&=\{(u',v,w)\in\{u'\}\times V_j\times V_k\:|\: (u',v)\in F'_{ij} \textrm{, }(u',w)\in F'_{ik}\textrm{ and }(v,w)\in F_{jk}\},
\end{align*}
and note that 
\[
|\Gamma_{ijk}\Delta\Gamma'_{ijk}|=|A|+|B|.
\]

Consider the set $C=\{v\in V_j\:|\:(u,v)\in F_{ij}\}$.
When $R_{t_1}$ is taken uniformly at random from 
$\Omega_{t_1}=\{T\subseteq \{1,\ldots,r_ir_j\}\mid |T|=f_{ij}\}$ 
(recall that $s_{t_1}=\{i,j\}$), 
the quantity $|C|$ has hypergeometric distribution 
$HG(r_ir_j,r_j,f_{ij})$.
By Lemma \ref{lemma:HGtail}(1), we have
\[
\Pr\left[|C|\ge 2 \frac{f_{ij}}{r_i}\right]\le \exp(-\frac{1}{3}\times \frac{f_{ij}}{r_i}).
\]
Let us fix $C$ and, for any $v\in C$, write
\begin{align*}
C(v)&=\{w\in V_k\:|\: (v,w)\in F_{jk}\}.
\end{align*}
When $R_{t_6}$ is taken uniformly at random from 
$\Omega_{t_6}=\{T\subseteq \{1,\ldots,r_jr_k\}\mid |T|=f_{jk}\}$ 
(recall that $s_{t_6}=\{j,k\}$), the 
quantity $|C(v)|$ has  hypergeometric distribution 
$HG(r_jr_k,r_k,f_{jk})$.
By Lemma \ref{lemma:HGtail}(1), we have
\begin{equation}\label{claim:eq1}
\Pr\left[|C(v)|\ge 2 \frac{f_{jk}}{r_j}\right]\le \exp(-\frac{1}{3}\times \frac{f_{jk}}{r_j}).
\end{equation}
Let us fix $C(v)$ and write
\begin{align*}
C'(v)&=\{w\in C(v)\:|\: (u,w)\in F_{ik}\}.
\end{align*}
When $R_{t_2}$ is taken 
uniformly at random from 
$\Omega_{t_2}=\{T\subseteq \{1,\ldots,r_ir_k\}\mid |T|=f_{ik}\}$ 
(recall that $s_{t_2}=\{i,k\}$), the quantity 
$|C'(v)|$ has hypergeometric distribution 
$HG(r_ir_k,|C(v)|,f_{ik})$.
Under the hypothesis $|C(v)|\le 2 \frac{f_{jk}}{r_j}$, 
we can 
apply Lemma \ref{lemma:HGtail}(3)
with $\delta=\frac{11f_{jk}/r_j}{|C(v)|}-1>2e-1$
to evaluate the size of $C'(v)$. 
The union bound then gives
\[
\Pr\left[|C'(v)|\le 11\frac{f_{jk}f_{ik}}{r_ir_jr_k}\right]\ge 1- 2^{-11\frac{f_{jk}f_{ik}}{r_ir_jr_k}}
-\exp(- \frac{f_{jk}}{3r_j}).
\]
Finally, note that $|A|=\sum_{v\in C}|C'(v)|$. Thus the union bound gives
\[
\Pr\left[|A|\le 22\frac{f_{ij}f_{ik}f_{jk}}{r^2_ir_jr_k}\right]\ge 1-\exp(-\frac{f_{ij}}{3r_i})-
2\frac{f_{ij}}{r_i}\left(
2^{-11\frac{f_{jk}f_{ik}}{r_ir_jr_k}}
+\exp(- \frac{f_{jk}}{3r_j})
\right).
\]
Similarly, we have
\[
\Pr\left[|B|\le 22\frac{f_{ij}f_{ik}f_{jk}}{r^2_ir_jr_k}\right]\ge 1-\exp(-\frac{f_{ij}}{3r_i})-
2\frac{f_{ij}}{r_i}\left(
2^{-11\frac{f_{jk}f_{ik}}{r_ir_jr_k}}
+\exp(- \frac{f_{jk}}{3r_j})
\right),
\]
and thus 
%\begin{align*}
\[
\Pr\left[
|\Gamma_{ijk}\Delta\Gamma'_{ijk}|\ge 44\times\frac{f_{ij}f_{ik}f_{jk}}{r_i^2r_jr_k}\right]\le
2\exp(-\frac{f_{ij}}{3r_i})+
%\\
4\frac{f_{ij}}{r_i}\left(
2^{-11\frac{f_{jk}f_{ik}}{r_ir_jr_k}}
+\exp(- \frac{f_{jk}}{3r_j})
\right),
\]
%\end{align*}
which is exponentially small since this set of parameters is admissible.
\end{proof}

Lemmas~\ref{lm:ApproxUnitary} and \ref{claim1} 
then show that the mapping
\begin{align*}
\ket{R_{t_4}}&\ket{R_{t_5}}\sum_{R_{t_1}\in \Omega_{t_1}}\sum_{R_{t_2}\in \Omega_{t_2}}\sum_{R_{t_6}\in \Omega_{t_6}}\sum_{R_{t_3}\in \Omega_{t_3}}
\ket{R_{t_1}}\ket{R_{t_2}}\ket{R_{t_6}}\ket{R_{t_3}}\ket{E_{ijk}(R_{t},R_{t_1},\cdots,R_{t_6})}\mapsto\\
&\hspace{-3mm}
\ket{R_{t_4}}\ket{R_{t_5}}\sum_{R_{t_1}\in \Omega_{t_1}}\sum_{R_{t_2}\in \Omega_{t_2}}\sum_{R_{t_6}\in \Omega_{t_6}}\sum_{R_{t_3}\in \Omega_{t_3}}
\ket{R_{t_1}}\ket{R_{t_2}}\ket{R_{t_6}}\ket{R_{t_3}}\ket{E_{ijk}(R'_{t},R_{t_1},\cdots,R_{t_6})}
\end{align*}
can be approximated within inverse polynomial precision 
using 
$
\tilde O(e_{ijk}/r_i+1)
$
queries. This argument is true for all $\{i,j,k\}\in \Sigma_3$, so the update cost is
\[
\mathsf{U}_t=\tilde O\Biggl(1+\sum_{\{j,k\}\textrm{ such that } \{i,j,k\}\in \Sigma_3}\frac{e_{ijk}}{r_i}\Biggr).
\]

Let us finally consider the case where $t_4,t_5,t_6$ are not all larger than $t$.
Whenever $t_6$ is larger than $t$, exactly the same analysis as above holds.
When $t_6$ is smaller than $t$ (which implies that $t_4$ and $t_5$ are also
smaller than $t$), remember that we only need to do the analysis of the update cost under 
the condition that $R_{t_6}$ is marked. This means that we can assume 
that, for any $v\in V_j$, we have $|\{w\in V_k\:|\: (v,w)\in F_{jk}\}|\le 2 f_{jk}/r_j$.
This property can be used instead of Inequality (\ref{claim:eq1}) in the proof
of Lemma \ref{claim1}, and the analysis then becomes the same as above. 

\vspace{2mm}

\noindent{\bf Case 2 \boldmath{[$s_t=\{i,j\}$ with $i<j$]}:} 
$R_t$ and $R'_t$ correspond to two subsets $F_{ij}$ and $F'_{ij}$ that
also differ by exactly one element: let us write $F'_{ij}=(F_{ij}\backslash \{(u,v)\})\cup\{(u',v')\}$. 
For any $\{i,j,k\}\in\Sigma_3$, there exist some $t_1,t_2\in\{1,\ldots,t-1\}$
such that $s_{t_1}=i$, $s_{t_2}=j$ and some $t_3\in\{t+1,\ldots,m\}$ such that 
$s_{t_3}=\{i,j,k\}$. There also exist some
$t_4,t_5,t_6\in\{1,\ldots,m\}$ such that $s_{t_4}=k$, $s_{t_5}=\{i,k\}$ and $s_{t_6}=\{j,k\}$.
Note that $t_4,t_5,t_6$ can be smaller than $t$, but we will first assume
that they are all larger than $t$ (the other cases, which are actually easier to analyze, 
are discussed at the end of the analysis). 
The states $(R_t,R_{t_1},R_{t_2},R_{t_4},R_{t_5},R_{t_6})$ define sets $F_{ij},V_i,V_j,V_k,F_{ik},F_{jk},\Gamma_{ijk}$, while the
states $(R'_t,R_{t_1},R_{t_2},R_{t_4},R_{t_5},R_{t_6})$ define sets $F'_{ij},V_i,V_j,V_k,F_{ik},F_{jk},\Gamma'_{ijk}$.
Given any state $R_{t_3}$,
let $E_{ijk}(R_{t},R_{t_1},\ldots,R_{t_6})$ denote the set of hyperedges to be queried associated with $\Gamma_{ijk}$ and $R_{t_3}$,
and $E_{ijk}(R'_{t},R_{t_1},\ldots,R_{t_6})$ denote the set of hyperedges to be queried associated with $\Gamma'_{ijk}$ and $R_{t_3}$.

By Lemmas~\ref{lm:ApproxUnitary} and \ref{lemma:construction}, 
we know that the mapping
\begin{align*}
\ket{R_{t_4}}&\ket{R_{t_5}}\ket{R_{t_6}}\sum_{R_{t_3}\in\Omega_{t_3}}\ket{R_{t_3}}\ket{E_{ijk}(R_{t},R_{t_1},R_{t_2},R_{t_3},R_{t_4},R_{t_5},R_{t_6})}
\mapsto\\
&\ket{R_{t_4}}\ket{R_{t_5}}\ket{R_{t_6}}\sum_{R_{t_3}\in\Omega_{t_3}}\ket{R_{t_3}}\ket{E_{ijk}(R'_{t},R_{t_1},R_{t_2},R_{t_3},R_{t_4},R_{t_5},R_{t_6})}
\end{align*}
can be approximated within inverse polynomial precision 
using 
\[
\tilde O\biggl(\frac{e_{ijk}|\Gamma_{ijk}\Delta\Gamma'_{ijk}|}{M_{ijk}}+\log n \biggr)=\tilde O\biggl(\frac{e_{ijk}|\Gamma_{ijk}\Delta\Gamma'_{ijk}|}{M_{ijk}}+1\biggr)
\] 
queries. 
We now prove the following lemma.
\begin{lemma}\label{claim2}
When $R_{t_5}$ and $R_{t_6}$ are taken uniformly at random,
\[
\Pr\left[
|\Gamma_{ijk}\Delta\Gamma'_{ijk}|\ge 22\times\frac{f_{ik}f_{jk}}{r_ir_jr_k}
\right]=O\biggl(\frac{1}{n^{100}}\biggr).
\]
\end{lemma}
\begin{proof}
Let us write
\begin{align*}
A&=\{w\in V_k\:|\:(u,w)\in F_{ik} \textrm{ and }(v,w)\in F_{jk}\},\\
B&=\{w\in V_k\:|\:(u',w)\in F_{ik} \textrm{ and }(v',w)\in F_{jk}\},
\end{align*}
and note that 
\[
|\Gamma_{ijk}\Delta\Gamma'_{ijk}|=|A|+|B|.
\]

Let us write $A_1=\{w\in V_k\:|\:(u,w)\in F_{ik}\}$.
When $R_{t_5}$ is taken uniformly at random from 
$\Omega_{t_5}=\{T\subseteq \{1,\ldots,r_ir_k\}\mid |T|=f_{ik}\}$ 
(recall that $s_{t_5}=\{i,k\}$), 
the quantity $|A_1|$ has hypergeometric distribution $HG(r_ir_k,r_k,f_{ik})$. 
By Lemma \ref{lemma:HGtail}(1), we have
\begin{equation}\label{claim:eq2}
\Pr\left[|A_1|\ge 2 \frac{f_{ik}}{r_i}\right]\le \exp(-\frac{1}{3}\times \frac{f_{ik}}{r_i}).
\end{equation}
Once $A_1$ is fixed, the quantity $|\{w\in A_1\:|\:(v,w)\in F_{jk}\}|$
has hypergeometric distribution $HG(r_jr_k,|A_1|,f_{jk})$ 
with expectation $\frac{|A_1|f_{jk}}{r_jr_k}$.
Under the assumption $|A_1|\le 2\frac{f_{ik}}{r_i}$,
we can apply Lemma \ref{lemma:HGtail}(3)
with $\delta=\frac{11f_{ik}/r_i}{|A_1|}-1>2e-1$. 
The union bound then gives
\[
\Pr\left[|A|\le 11\frac{f_{ik}f_{jk}}{r_ir_jr_k}\right]\ge 1- 2^{-11\frac{f_{ik}f_{jk}}{r_ir_jr_k}}-\exp(-\frac{f_{ik}}{3r_i}).
\]
Similarly we obtain 
\[
\Pr\left[|B|\le 11\frac{f_{ik}f_{jk}}{r_ir_jr_k}\right]\ge 1- 2^{-11\frac{f_{ik}f_{jk}}{r_ir_jr_k}}-\exp(-\frac{f_{ik}}{3r_i}),
\]
and thus 
\[
\Pr\left[
|\Gamma_{ijk}\Delta\Gamma'_{ijk}|\ge 22\frac{f_{ik}f_{jk}}{r_ir_jr_k}\right]\le
2\times \left(2^{-11\frac{f_{ik}f_{jk}}{r_ir_jr_k}}+\exp(-\frac{f_{ik}}{3r_i})\right),
\]
which is exponentially small since 
this set of  parameters is admissible.
%\qed
\end{proof}

Using Lemma~\ref{lm:ApproxUnitary} and Lemma \ref{claim2},
the mapping
\begin{align*}
\ket{R_{t_4}}&\sum_{R_{t_5}\in \Omega_{t_5}}\sum_{R_{t_6}\in \Omega_{t_6}}\sum_{R_{t_3}\in \Omega_{t_3}}
\ket{R_{t_5}}\ket{R_{t_6}}\ket{R_{t_3}}\ket{E_{ijk}(R_{t},R_{t_1},\cdots,R_{t_6})}\mapsto\\
&
\ket{R_{t_4}}\sum_{R_{t_5}\in \Omega_{t_5}}\sum_{R_{t_6}\in \Omega_{t_6}}\sum_{R_{t_3}\in \Omega_{t_3}}
\ket{R_{t_5}}\ket{R_{t_6}}\ket{R_{t_3}}\ket{E_{ijk}(R'_{t},R_{t_1},\cdots,R_{t_6})}
\end{align*}
can be approximated within inverse polynomial precision 
using 
$
\tilde O(e_{ijk}/f_{ij}+1)
$
queries. This argument is true for all $\{i,j,k\}\in \Sigma_3$, so the update cost is
\[
\mathsf{U}_t=\tilde O\Biggl(1+\sum_{k\textrm{ such that }\{i,j,k\}\in \Sigma_3}\frac{e_{ijk}}{f_{ij}}\Biggr).
\]

Let us finally consider the case where $t_4,t_5,t_6$ are not all larger than $t$.
Whenever both $t_5$ and $t_6$ are larger than $t$, exactly the same analysis as above holds.
When $t_5<t<t_6$, remember that we only need to do the analysis of the update cost under 
the condition that $R_{t_5}$ is marked. This means that we can assume 
that, for any $u\in V_i$, we have $|\{w\in V_k\:|\: (u,w)\in F_{ik}\}|\le 2 f_{ik}/r_i$.
This property can be used instead of Inequality~(\ref{claim:eq2}) in the proof
of Lemma \ref{claim2}, and the analysis then becomes the same as above. 
When $t_6<t<t_5$, the same argument holds by inverting the roles 
of $\{i,k\}$ and $\{j,k\}$ in the proof of Lemma~\ref{claim2}. 
When $t_5,t_6<t$, the fact that $R_{t_5}$ and $R_{t_6}$ are marked 
(more precisely, item (d) in the definition of marked states of Section~\ref{subsection:marked})
implies that
for any $(u_1,v_1)\in V_i\times V_j$,
$|\{w\in V_k\:|\:(u_1,w)\in F_{ik} \textrm{ and }(v_1,w)\in F_{jk}\}|\le 11\frac{f_{ik}f_{jk}}{r_ir_jr_k}$,
which immediately implies that 
$|\Gamma_{ijk}\Delta\Gamma'_{ijk}|\le 22\frac{f_{ik}f_{jk}}{r_ir_jr_k}$.
\vspace{2mm}

\noindent{\bf Case 3 \boldmath{[$s_t=\{i,j,k\}$ with $i<j<k$]}:} 
%\noindent
$R_t$ and $R'_t$ are two subsets of $\{1,\ldots,M_{ijk}\}$, both of size $e_{ijk}$,
differing by exactly one element. The corresponding $E_{ijk}$ and $E'_{ijk}$ are subsets of 
the same $\Gamma_{ijk}$, and have symmetric difference $|E_{ijk}\Delta E'_{ijk}|\le 2$, so $\mathsf{U_t}\le 2$.
\vspace{2mm}

Now the proof of Theorem \ref{th:main} is completed.

%%%%%%%%%%%%%%%%%%%%%
\section{Applications: 4-clique detection and ternary associativity testing
%associativity
}\label{sec:applications}
%%%%%%%%%%%%%%%%%%%%%
In this section we describe how to use our method to construct efficient 
quantum algorithms for 4-clique detection and ternary associativity testing.

First, by applying Theorem \ref{th:main} to the case where $H$ is a $4$-clique,
and optimizing both the loading schedule and the parameters, we obtain the following 
result. 
\begin{theorem}\label{th:4-clique}
There exists a quantum algorithm that detects if
a 3-uniform hypergraph on $n$ vertices has a 4-clique, 
with high probability, 
using $\tilde O(n^{241/128})=O(n^{1.883})$ queries. 
\end{theorem}
\begin{proof}
We use Theorem \ref{th:main}. Among the $1680384$ possible valid loading schedules, 
we found, 
by numerical search, that the best schedule is
\[
(1,2,3,4,\{1,2\},\{1,3\},\{1,4\},\{2,3\},\{2,4\},\{3,4\},\{1,2,3\},\{1,2,4\},\{1,3,4\}, \{2,3,4\}).
\]
%\begin{align*}
%(1,2,3,4,\{1,2\},\{1,3\},\{1,4\},&\{2,3\},\{2,4\},\{3,4\},\{1,2,3\},\\
%&\{1,2,4\},\{1,3,4\}, \{2,3,4\}).
%\end{align*}
The complexity of the algorithm for this schedule is minimized by the following  values of parameters: 
\begin{align*}
&r_1=n^{1/2},  &  &r_2=n^{3/4}, &  &r_3=n^{7/8},  &         &r_4=n^{3/4}, \\
&f_{12}=n^{5/4}, & &f_{13}=n^{5/4}, & &f_{14}=n^{147/128}, & & \\
&f_{23}=n^{193/128}, & &f_{24}=n^{83/64}, & &f_{34}=n^{181/128}, & & \\
&e_{123}=n^{241/128}, & &e_{124}=n^{217/128}, & &e_{134}=n^{211/128}, & &e_{234}=n^{193/128}.
\end{align*}
It is easy to check that this set of parameters is admissible.
This gives query complexity $\tilde O(n^{241/128})$.
\end{proof}

Next, we consider ternary associativity testing.
Let $X$ be a finite set with ${|X|=n}$. 
A {ternary operator} ${\cal F}$ from ${X\times X\times X}$ to $X$ 
is said to be {\em associative} if ${\cal F}({\cal F}(a,b,c),d,e)={\cal F}(a,{\cal F}(b,c,d),e)
={\cal F}(a,b,{\cal F}(c,d,e))$ holds for every 5-tuple $(a,b,c,d,e)\in X^5$.
The function ${\cal F}$ is given as a black-box: when we make a query $(a,b,c)$ 
to ${\cal F}$, the answer ${\cal F}(a,b,c)$ is returned. 
We can show that 
that the property
``$\cal F$ is not associative" has a certificate corresponding to a sub-hypergraph of 
seven vertices in a 3-uniform directed hypergraph
with each edge weighted by an element in $X$.
By applying Theorem \ref{th:main}
with adaptations to directed hypergraphs with non-binary hyperedge weights,
we obtain the following result.

\begin{theorem}\label{thm:associativity}
There exists a quantum algorithm that determines if ${\cal F}$ is associative 
with high probability using $\tilde O(n^{169/80})=\tilde O(n^{2.1125})$ queries. 
\end{theorem}
\begin{proof}
To apply Theorem~\ref{th:main}, we basically follow the approach of Ref.~\cite{Lee+SODA13} 
for the (binary) associativity testing. If ${\cal F}$ is not associative, there is a 5-tuple 
$(a_1,a_2,a_3,a_4,a_5)\in X^5$ such that (i) ${\cal F}({\cal F}(a_1,a_2,a_3),a_4,a_5)
\neq {\cal F}(a_1,{\cal F}(a_2,a_3,a_4),a_5)$ 
or (ii)  ${\cal F}(a_1,{\cal F}(a_2,a_3,a_4),a_5)\neq {\cal F}(a_1,a_2,{\cal F}(a_3,a_4,a_5))$. 
Thus, it suffices to check case (i) and case (ii) individually.

We consider only case (i) since case (ii) is similarly analyzed and needs 
the same query complexity as the algorithm for case (i).
A certificate to case (i) is given by a~7-tuple $(a_1,a_2,\ldots,a_7)\in X^7$ 
such that ${\cal F}(a_1,a_2,a_3)=a_6$, ${\cal F}(a_2,a_3,a_4)=a_7$ 
and ${\cal F}(a_6,a_4,a_5)\neq {\cal F}(a_1,a_7,a_5)$. 
Let $H$ be a directed hypergraph on seven vertices 
with directed hyperedges $(1,2,3),(2,3,4),(6,4,5),(1,7,5)$.   
Then, finding a certificate to case~(i) can be reduced to 
finding a sub-hypergraph isomorphic to $H$ in an $n$-vertex directed hypergraph with each hyperedge weighted with an element in $X$,
to which we will apply Theorem \ref{th:main}.
Note that, although the proof of Theorem \ref{th:main} assumes the given hypergraph is undirected
and each hyperedge is weighted with binary values, 
we can easily adapt the algorithm to handle directed hypergraphs
with non-binary hyperedge weight:
(1) to deal with directed hyperedges of $H$ we simply replace a query 
to the black-box on an unordered triple by a query on the corresponding 
ordered triple
(for instance, for $(u,v,w)\in V_4\times V_5\times V_6$ we 
will query $\chi((w,u,v))$ instead of $\chi(\{u,v,w\})$);
(2) since the quantum walk actually does not use the property
that hyperedges have binary weight, it works without modification 
for the case
of non-binary hyperedge weights as well.
Note also that the resulting algorithm searches $H$ over $X^7$,
so we do not need to consider separately the case of detecting vertex contractions of $H$ as in Ref.~\cite{Lee+SODA13}.

By numerical search, we found the following schedule: 
\begin{align*}
(1,3,4,6,2,5,7,&\{1,2\},\{1,3\},\{1,5\},\{1,7\},\{2,3\},\{2,4\},\{3,4\},\{4,5\},\{4,6\},\\ 
&\{5,6\},\{5,7\},
\{1,2,3\},\{1,5,7\},\{2,3,4\},\{4,5,6\}).
\end{align*} 
The complexity of the algorithm for this schedule is minimized 
by the following values of parameters: 
$r_1=n^{3/4}$, $r_2=n$, $r_3=n$, $r_4=n^{7/8}$, $r_5=n^{1/2}$, $r_6=n$, $r_7=n$;
$f_{12}=n^{7/4}$, $f_{13}=n^{7/4}$, $f_{15}=n^{5/4}$, $f_{17}=n^{7/4}$, $f_{23}=n^{23/16}$, 
$f_{24}=n^{29/16}$, $f_{34}=n^{15/8}$, $f_{45}=n^{11/8}$, $f_{46}=n^{15/8}$, $f_{56}=n^{3/2}$,
$f_{57}=n^{3/2}$;
$e_{123}=n^{169/80}$, $e_{157}=n^{169/80}$, $e_{234}=n^{169/80}$ and $e_{456}=1$. 
It is easy to check that this set of parameters is admissible.
This gives query complexity $\tilde O(n^{169/80})$. 
\end{proof}

\section*{Acknowledgments}
The authors are grateful to Tsuyoshi Ito, Akinori Kawachi, Hirotada Kobayashi,
Masaki Nakanishi, Masaki Yamamoto and Shigeru Yamashita for helpful discussions.
This work is supported by the Grant-in-Aid for Scientific Research~(A)~No.~24240001 of the JSPS
and the Grant-in-Aid for Scientific Research on Innovative Areas~No.~24106009 of
the MEXT in Japan.

\end{document}